\documentclass[11pt]{article}
\usepackage{latexsym}
\usepackage{mathrsfs}
\usepackage{amsmath,amsthm}
\usepackage{graphicx,epstopdf,epsfig,multirow,epic,bm}
\usepackage{color,fancyvrb}
\usepackage{colortbl}
\usepackage{multicol}
\usepackage{amssymb}

\usepackage{esint}  

\usepackage{makeidx}
\usepackage{showidx}
\usepackage{textcomp}
\usepackage{paralist}

\usepackage{CJK}

\theoremstyle{plain}
\newtheorem{thm}{Theorem}
\newtheorem{cor}[thm]{Corollary}

\theoremstyle{definition}
\newtheorem{defn}{Definition}
\theoremstyle{plain}

\theoremstyle{problem}

\theoremstyle{plain}

\theoremstyle{plain}

\theoremstyle{plain}

\DeclareGraphicsExtensions{.eps,.eps.gz}
\normalsize \rm
\begin{document}

\thispagestyle{empty}

\begin{center}
{\huge\bf Improper Graceful and Odd-graceful Labellings of Graph Theory}\\[8pt]
\end{center}
\begin{center}
{\large Hongyu Wang$^{a,b}$\quad Jin Xu$^{a}$\quad Bing Yao$^{b,}$\footnote{Corresponding author, Email: yybb918@163.com}}\\[8pt]
{\footnotesize a. School of Electronics Engineering and Computer Science, Peking University, 100871, China \\
b. College of Mathematics and Statistics, Northwest
Normal University, Lanzhou, 730070, China}
\end{center}

\begin{abstract}
In this paper we define some new labellings for trees, called the  in-improper and out-improper odd-graceful labellings such that some trees labelled with the new labellings can induce graceful graphs having at least a cycle. We, next, apply the new labellings to construct large scale of graphs having improper graceful/odd-graceful labellings or having graceful/odd-graceful labellings.\\[6pt]
\end{abstract}

\section{Introduction and concepts}

Graph colorings/labellings of graph theory can be applied in many areas of science such as bioinformatics, (scale-free, small-world) networks, VLSI. As known, colored/labelled graphs are used in many areas of today's science and computer science. For example, some design algorithms are related with such graphs; the problem of frequency distribution of communication can be investigated by graph distinguishing colorings. We try to explore connection between labelled graphs and cryptographs of information security for graphic cryptograph. The survey article \cite{Gallian2012} in which more than 1400 papers on various graph labellings were collected provides a good overview of concepts and results of current graph labellings. It has been known that all lobsters are odd-graceful (Ref. \cite{Zhou-Yao-Chen-Tao2012}), however ``\emph{all lobsters are graceful}'' conjectured by Bermond \cite{J-C-Bermond1979} is open up now. In \cite{Van-Bussel-2002}, the author investigated the so-called \emph{relaxed graceful labelling} of trees. We are motivated from the relaxed graceful labelling, and then propose some improper graceful labellings for trees in order to producing graceful graphs having at least a cycle.

We use standard terminology and notation of graph theory,  and all graphs mentioned here are simple, undirected and finite. For the sake of simplicity, the shorthand symbol ${[m, n]}$ stands for a set $\{m, m+1, \dots, n\}$,
where $m$ and $n$ are integers with $0\leq m< n$; the notation ${[s, t]^{o}}$ indicates an odd-set $\{s, s+2, \dots, t\}$, where $s$ and $t$ both are odd integers with $1\leq s< t$. A \emph{leaf} is a vertex of degree one, and $L(G)$ stands for the set of leaves of a graph $G$. A \emph{caterpillar} is a tree $T$ such that the deletion of leaves of $T$ results in a path. A \emph{lobster} is a tree  such that the graph obtained by removing all leaves of the tree is just a
caterpillar. A $(p,q)$-graph is a graph having $p$ vertices and $q$ edges.

\begin{defn} \label{defn:proper-bipartite-labelling-ongraphs}
Suppose that a connected $(p,q)$-graph $G$ admits a mapping
$\theta:V(G)\rightarrow \{0,1,2,\dots \}$. For every edge $xy\in E(G)$
the induced edge label is defined as
$\theta(xy)=|\theta(x)-\theta(y)|$, and write
$\theta(V(G))=\{\theta(u):u\in V(G)\}$,
$\theta(E(G))=\{\theta(xy):xy\in E(G)\}$. There are the following
constraints:
\begin{asparaenum}[(a)]
\item \label{Proper01} $|\theta(V(G))|=p$.
\item \label{Graceful-001} $\theta(V(G))\subseteq [0,q]$, $\min \theta(V(G))=0$.
\item \label{Odd-graceful-001} $\theta(V(G))\subset [0,2q-1]$, $\min \theta(V(G))=0$.
\item \label{Graceful-002} $\theta(E(G))=[1,q]$.
\item \label{Odd-graceful-002} $\theta(E(G))=[1,2q-1]^o$.
\item \label{Set-ordered} $G$ is a bipartite graph with the bipartition
$(X,Y)$ such that $\max\{\theta(x):x\in X\}< \min\{\theta(y):y\in
Y\}$ ($\theta(X)<\theta(Y)$ for short).
\item \label{Graceful-matching} $G$ is a tree containing a perfect matching $M$ such that
$\theta(x)+\theta(y)=q$ for each edge $xy\in M$.
\item \label{Odd-graceful-matching} $G$ is a tree having a perfect matching $M$ such that
$\theta(x)+\theta(y)=2q-1$ for each edge $xy\in M$.
\end{asparaenum}

Then $\theta$ is a \emph{graceful labelling} if it holds (\ref{Proper01}), (\ref{Graceful-001}) and (\ref{Graceful-002});  $\theta$ is a \emph{set-ordered graceful labelling} if it holds (\ref{Proper01}),
(\ref{Graceful-001}), (\ref{Graceful-002}) and (\ref{Set-ordered});  $\theta$ is 
a \emph{strongly graceful labelling} if it holds (\ref{Proper01}),
(\ref{Graceful-001}), (\ref{Graceful-002}) and
(\ref{Graceful-matching});  $\theta$ is a \emph{strongly set-ordered graceful
labelling} if it holds (\ref{Proper01}), (\ref{Graceful-001}),
(\ref{Graceful-002}), (\ref{Set-ordered}) and
(\ref{Graceful-matching});  $\theta$ is an \emph{odd-graceful labelling} if it holds (\ref{Proper01}),
(\ref{Odd-graceful-001}) and (\ref{Odd-graceful-002});  $\theta$ is a
\emph{set-ordered odd-graceful labelling} if it holds
(\ref{Proper01}), (\ref{Odd-graceful-001}), (\ref{Odd-graceful-002})
and (\ref{Set-ordered});  $\theta$ is a \emph{strongly odd-graceful labelling}
if it holds (\ref{Proper01}), (\ref{Odd-graceful-001}),
(\ref{Odd-graceful-002}) and (\ref{Odd-graceful-matching});  $\theta$ is a
\emph{strongly set-ordered odd-graceful labelling} if it holds
(\ref{Proper01}), (\ref{Odd-graceful-001}),
(\ref{Odd-graceful-002}), (\ref{Set-ordered}) and
(\ref{Odd-graceful-matching}).
\end{defn}

Also, graceful/odd-graceful labellings are called \emph{proper graceful/odd-graceful labellings} in the following discussion.

\begin{defn} \label{defn:2-improper-graceful-labelling}
Let $G$ be a $(p,q)$-graph, and $f$ be a mapping from $V(G)$ to $\{0,1,2,\dots \}$ such that $\min f(V(G))=0$ and $\max f(V(G))=m$, where $f(V(G))=\{f(x):x\in V(G)\}$. Write $f(E(G))=\{f(uv)=|f(u)-f(v)|:uv\in E(G)\}$. There are four restrictions:

C1. $f(E(G))=[1,q]$ and $m\leq q$; C2. $f(E(G))=[1,q]$ and $m>q$;

C3. $f(E(G))=[1,2q-1]^o$ and $m\leq 2q-1$; C4. $f(E(G))=[1,2q-1]^o$ and $m>2q-1$.\\
Then, $f$ is called an \emph{in-improper graceful labelling} (in-imgl) if $f$ holds C1, an \emph{out-improper graceful labelling} (out-imgl) if $f$ holds C2, an \emph{in-improper odd-graceful labelling} (in-imoddgl) if $f$ holds C3, and an \emph{out-improper odd-graceful labelling} (out-imoddgl) if $f$ holds C4.
\end{defn}

The result in Theorem \ref{lemma:basic-theorem-000} have been shown in some literature (Ref. \cite{Gallian2012}). We show a short proof of Theorem \ref{lemma:basic-theorem-000} here.

\begin{thm} \label{lemma:basic-theorem-000}
Every tree admits $(i)$ in-imgls or out-imgls; and $(ii)$ in-imoddgls or out-imoddgls.
\end{thm}
\begin{proof} We present the proof of the assertion $(i)$ by introduction on orders of trees, since the proof of the assertion $(ii)$ is very similar with that of the assertion $(i)$.

Let $T$ be a tree on $n$ vertices, so $T$ has $(n-1)$ edges. For $n=2$, $T$ admits a graceful labelling. Suppose that the theorem  holds true for trees having smaller numbers of vertices. We consider a tree $H=T-v$, where $v$ is a leaf of $T$ and is adjacent to $u$ in $T$. So $H$ has an in-imgl or an out-imgl $f$ such that $f(E(H))=[1,n-2]$ by induction hypothesis. We label the vertex $v$ of $T$ as $n-1+f(u)$. Thereby, we obtain an out-imgl $g$ as: $g(v)=n-1+f(u)$, $g(x)=f(x)$ for $x\in V(H)=V(T)\setminus \{v\}$. Clearly, $g(E(T))=[1,n-1]$. The proof of the theorem  is complete.
\end{proof}

Based on Theorem \ref{lemma:basic-theorem-000} there are two parameters under a labelling $f$ of a tree $T$ on $n$ vertices: (1) $k=|S|$, where $S$ is a maximum subset of $V(T)$ such that every vertex $u\in S$ holds $f(u)=f(u\,')$ for some $u\,'\in S\setminus\{u\}$; (2) $l=n-c$, where $c$ is the number of distinct labels assigned to the vertices of $T$. We call $f$ an \emph{in-$(k,l)$-imgl} (resp. \emph{out-$(k,l)$-imgl}) as if $f$ is an in-imgl (resp. out-imgl) of $T$, and furthermore we call $f$ a \emph{proper graceful labelling} if $(k,l)=(0,0)$; and we call $f$ an \emph{in-$(k,l)$-imoddgl}  (resp. \emph{out-$(k,l)$-imoddgl}) if  $f$ is an in-imoddgl (resp. out-imoddgl) of $T$, and furthermore we call $f$ a \emph{proper odd-graceful labelling} when $(k,l)=(0,0)$. There exist the following possible parameters:

$l^{\max}_{in}=\max\{l:~\textrm{in-$(k,l)$-imgls}\}$ and $k^{\max}_{in}=\max \{k:~\textrm{in-$(k,l^{\max}_{in})$-imgls}\}$;

 $l^{\min}_{in}=\max\{l:~\textrm{in-$(k,l)$-imgls}\}$ and $k^{\min}_{in}=\max \{k:~\textrm{in-$(k,l^{\min}_{in})$-imgls}\}$;

$l^{\min}_{out}=\min\{l:~\textrm{out-$(k,l)$-imgls}\}$ and  $k^{\min}_{out}=\min\{l:~\textrm{out-$(k,l^{\min}_{out})$-imgls}\}$.

Similarly, it may be interesting to consider: \emph{Determine in-$(k^{\max}_{in},l^{\max}_{in})$-imoddgls and out-$(k^{\min}_{in}, l^{\min}_{in})$-imoddgls of trees}, where $l^{\max}_{in}=\max\{l:~\textrm{in-$(k,l)$-imoddgls}\}$, $k^{\max}_{in}=\max \{k:~\textrm{in-$(k,l^{\max}_{in})$-imoddgls}\}$, $l^{\min}_{in}=\min\{l:~\textrm{in-$(k,l)$-imoddgls}\}$ and $k^{\min}_{in}=\min \{k:~\textrm{in-$(k,l^{\min}_{in})$-imoddgls}\}$.

\begin{figure}[h]
\centering
\includegraphics[height=5cm]{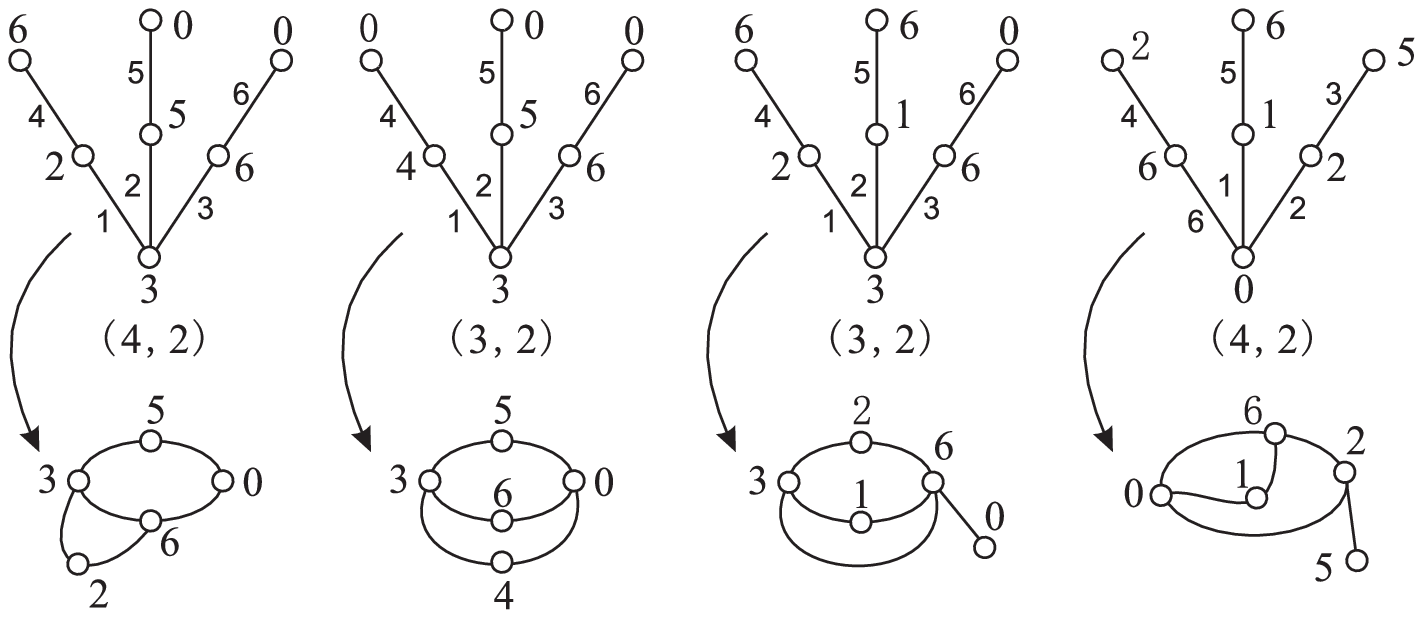}
\caption{\label{fig:definition-111}  {\small  A lobster having four in-improper graceful labellings induces four graceful graphs.}}
\end{figure}

Rosa \cite{Rosa1966} in 1966 conjectured: \emph{Every tree is graceful}. Thereby, Rosa's conjecture implies: Every tree admits in-improper graceful labellings, and furthermore every in-$(k^{\min}_{in},l^{\min}_{in})$-imgl is equal to a graceful labelling in a tree. In 1991 Gnanajothi \cite{R-B-Gnanajothi1991} gave a conjecture: \emph{Every tree is odd-graceful}, which leads that every tree admits in-improper odd-graceful labellings; and furthermore each in-$(k^{\min}_{in},l^{\min}_{in})$-imoddgl is an odd-graceful labelling in a tree.

\section{Finding in-imgls and in-imoddgls of graphs}

The term ``adding leaves to a $(p,q)$-graph $G$'' means that we join new vertices $u_{i,1},u_{i,2},\dots ,u_{i,m_i}$ to each vertex $u_i$ of the $(p,q)$-graph $G$ for $i\in [1,p]$, here, some $m_i$ may be zero.

\begin{thm} \label{them:basic-theorem}
Suppose that a tree $T$ admit set-ordered graceful/odd-graceful labellings. Then adding leaves to $T$ produces a tree that admits in-imgls and proper odd-graceful labellings.
\end{thm}
\begin{proof} A tree $T$ has its own bipartition $V(T)=X\cup Y$ with $X=\{x_1,x_2,\dots,x_s\}$ and $Y=\{y_1,y_2,\dots,y_t\}$ such that $s+t=|T|$, $u\in X$ and $v\in Y$ for each edge $uv\in E(T)$. Adding leaves to $T$ produces a tree $H$. So, we have that every vertex $x_i$ of the tree $H$ is added the leaves $x_{i,1},x_{i,2},\dots ,x_{i,m_i}$ for $i\in [1,s]$, and every vertex $y_j$ of the tree $H$ is added the leaves $y_{j,1},y_{j,2},\dots ,y_{j,n_j}$ for $j\in [1,t]$. Here, it may be that some $m_i=0$ or $n_j=0$. Let $M_x=\sum^{s}_{k=1}m_k$ and $M_y=\sum^{t}_{k=1}n_k$.

First, we consider the case: $T$ has a set-ordered graceful labelling $f$ such that $0=f(x_1)\leq f(x_i)<f(x_{i+1})$ for $i\in [1,s-1]$, $f(x_s)<f(y_1)$, and $f(y_i)<f(y_{i+1})\leq f(y_t)=s+t-1$ for $i\in [1,t-1]$. We define a labelling $g$ of the tree $H$ as follows.

\emph{Step 1. }We label the edges of $H$ in the following way: $g(x_1x_{1,a})=a$ for $a\in [1,m_1]$, $g(x_ix_{i,a})=a+\sum^{i-1}_{k=1}m_k$ for $a\in [1,m_i]$ and $i\in [2,s]$; $g(y_ty_{t,b})=M_x+b$ for $b\in [1,n_t]$, $g(y_{t-j+1}y_{t-j+1,b})=M_x+b+\sum^{j-1}_{k=1}n_{t-k+1}$ for $b\in [1,n_j]$ and $j\in [2,t]$; $g(y_1y_{1,b})=M_x+M_y-n_1+b$ for $b\in [1,n_1]$; $g(x_iy_j)=f(x_iy_j)+M_x+M_y$ for $x_iy_j\in E(T)$.

\emph{Step 2. }We label the vertices of the tree  $H$ by setting $g(x_i)=f(x_i)$ for $i\in [1,s]$; $g(y_j)=f(y_j)+M_x+M_y$ for $j\in [1,t]$; $g(x_{i,a})=g(x_i)+g(x_ix_{i,a})$ for $a\in [1,m_i]$ and $i\in [1,s]$; $g(y_{j,b})=g(y_j)-g(y_jy_{j,b})$ for $b\in [1,n_j]$ and $j\in [1,t]$. Notice that $g(x_i)<g(x_{i,a})=g(x_i)+g(x_ix_{i,a})\leq f(x_s)+\sum^{i}_{k=1}m_k$ for $a\in [1,m_i]$ and $i\in [1,s]$; and  for $b\in [1,n_j]$ and $j\in [1,t]$, we then  have
$${
\begin{split}
g(y_j)>g(y_{j,b})&=g(y_j)-g(y_jy_{j,b})\\
&=f(y_j)+M_x+M_y-\left (M_x+b+\sum^{j-1}_{k=1}n_{t-k+1}\right )\\
&\geq f(y_j)+\sum^{t-j}_{k=1}n_{k}.
\end{split}}
$$

\emph{Step 3. }Notice that $0=g(x_1)<g(u)\leq g(y_t)=s+t-1+M_x+M_y=|H|-1$ for $u\in V(H)$, and
$${
\begin{split}
g(E(H))&=\{g(x_ix_{i,a}): a\in [1,m_i],i\in [1,s]\}\cup \{g(y_jy_{j,b}):b\in [1,n_j],j\in [1,t]\}\cup g(E(T))\\
&=[1,|E(H)|-1].
\end{split}}
$$ It follows that $g$ is an in-imgl of the tree  $H$.

Second, to prove that the tree  $H$ admits an in-imoddgl, we take a set-ordered odd-graceful labelling $\alpha$ of $T$. By the definition of a set-ordered odd-graceful labelling, we have $0=\alpha(x_1)\leq \alpha(x_i)<\alpha(x_{i+1})$ for $i\in [1,s-1]$, $\alpha(x_s)<\alpha(y_1)$, and  $\alpha(y_i)<\alpha(y_{i+1})\leq \alpha(y_t)=2|T|-3$ for $i\in [1,t-1]$. Since $\alpha(E(T))=[1,2|T|-3]^o$, so the vertices of $X$ have the same parity, so do the vertices of $Y$. Without loss of generality, we assume that each $\alpha(x_i)$ is even, and each $\alpha(y_j)$ is odd. We define another labelling $\beta$ of the tree  $H$ as follows.

We, now, label the edges of the tree  $H$ by defining $\beta(x_1x_{1,a})=2a-1$ for $a\in [1,m_1]$, $\beta(x_ix_{i,a})=2a-1+2\sum^{i-1}_{k=1}m_k$ for $a\in [1,m_i]$ and $i\in [2,s]$; $\beta(y_ty_{t,b})=2M_x-1+2b$ for $b\in [1,n_t]$, $\beta(y_{t-j+1}y_{t-j+1,b})=2M_x-1+2b+2\sum^{j-1}_{k=1}n_{t-k+1}$ for $b\in [1,n_j]$ and $j\in [2,t]$; $\beta(y_1y_{1,b})=2(M_x+M_y-n_1+b)-1$ for $b\in [1,n_1]$; $\beta(x_iy_j)=\alpha(x_iy_j)+2(M_x+M_y)$ for $x_iy_j\in E(T)$.

We, next, label the vertices of the tree  $H$ by setting $\beta(x_i)=\alpha(x_i)$ for $i\in [1,s]$; $\beta(y_j)=\alpha(y_j)+2(M_x+M_y)$ for $j\in [1,t]$;  $\beta(x_{i,a})=\beta(x_i)+\beta(x_ix_{i,a})$ for $a\in [1,m_i]$ and $i\in [1,s]$ (here $\beta(x_{i,a})$ is odd);
$\beta(y_{j,b})=\beta(y_j)-\beta(y_jy_{j,b})$ for $b\in [1,n_j]$ and $j\in [1,t]$ (here $\beta(y_{j,b})$ is even).

Notice that $\beta(x_i)<\beta(x_{i,a})=\beta(x_i)+\beta(x_ix_{i,a})\leq \alpha(x_s)-1+2\sum^{i}_{k=1}m_k$ for $a\in [1,m_i]$ and $i\in [1,s]$. For $b\in [1,n_j]$ and $j\in [1,t]$, we so have
$${
\begin{split}
\beta(y_j)>\beta(y_{j,b})&=\beta(y_j)-\beta(y_jy_{j,b})\\
&=\alpha(y_j)+2(M_x+M_y)-\left (2M_x-1+2b+2\sum^{j-1}_{k=1}n_{t-k+1}\right )\\
&\geq \alpha(y_j)-1+2\sum^{t-j}_{k=1}n_{k}.
\end{split}}
$$ So, no two vertices $u$ and $v$ of $H$ hold $\beta(u)=\beta(v)$. It follows
$${
\begin{split}
\beta(E(H))&=\{\beta(x_ix_{i,a}): a\in [1,m_i],i\in [1,s]\}\cup \{\beta(y_jy_{j,b}):b\in [1,n_j],j\in [1,t]\}\cup \beta(E(T))\\
&=[1,2|E(H)|-1]^o
\end{split}}
$$
that $g$ is a proper odd-graceful labelling of the tree  $H$. The theorem is covered.
\end{proof}

For a lobster $T$, we have a caterpillar $T-L(T)$ by the definition of a lobster. Notice that any caterpillar admits set-ordered graceful/odd-graceful labellings \cite{Zhou-Yao-Chen-Tao2012}. So, the following Corollary \ref{cor:Corollary11} is an immediate consequence of Theorem  \ref{them:basic-theorem}.

\begin{cor} \label{cor:Corollary11}
Every lobster admits in-imgls and proper odd-graceful labellings.
\end{cor}

\begin{thm} \label{thm:1233333}
Let $G$ be a connected bipartite graph having vertices $w_1,w_2,\dots ,w_n$. For each fixed integer $i\in [1,n]$, identifying a certain vertex of every connected bipartite graph $G_{i,j}$ for $j\in [1,m_i]$ (it is allowed some $m_i=0$) with $w_i$ into a vertex produces a connected graph $H^*$. Then $H^*$ admits an in-imgl (resp. in-imoddgl) if $G$ and every $G_{i,j}$ admit set-ordered graceful labellings (resp. set-ordered odd-graceful labellings).
\end{thm}
\begin{proof}We show the proof of the first case ``\emph{$H^*$ admits an in-imgl if $G$ and every $G_{i,j}$ admit set-ordered graceful labellings}'', and the proof for the second case ``\emph{$H^*$ admits an in-imoddgl if $G$ and every $G_{i,j}$ admit set-ordered odd-graceful labellings}'' can be omitted since it is the same as that of the first case.

By the theorem's hypothesis, the connected bipartite graph $G$ has its own bipartition $V(G)=X\cup Y$ with $X=\{u_i:i\in [1,s]\}$ and $Y=\{v_j:j\in [1,t]\}$, and has a set-ordered graceful labelliing $f$ such that $0=f(u_1)\leq f(u_i)<f(u_{i+1})$ for $i\in [1,s-1]$, $f(u_s)<f(v_1)$, and $f(v_i)<f(v_{i+1})\leq f(v_{t})=|E(G)|$ for $i\in [1,t-1]$.

For a fixed integer $i\in [1,s]$, we will identify a certain vertex of every connected bipartite $(p^u_{i,j},q^u_{i,j})$-graph $G^u_{i,j}$ with the vertex $u_i\in X$ into a vertex, $j\in [1,m_i]$. Let $V(G^u_{i,j})=X^u_{i,j}\cup Y^u_{i,j}$ be the bipartition of $V(G^u_{i,j})$, where $X^u_{i,j}=\{x^u_{i,j,l}:l\in [1,m^u_{i,j}]\}$ and $Y^u_{i,j}=\{y^u_{i,j,r}:r\in [1,n^u_{i,j}]\}$; $G^u_{i,j}$ has a set-ordered graceful labelliing $f_{i,j}$ such that $0=f_{i,j}(x^u_{i,j,1})\leq f_{i,j}(x^u_{i,j,l})<f_{i,j}(x^u_{i,j,l+1})$ for $l\in [1,m^u_{i,j}-1]$, $f_{i,j}(x^u_{i,j,m^u_{i,j}})<f_{i,j}(y^u_{i,j,1})$, $f_{i,j}(y^u_{i,j,r})<f_{i,j}(y^u_{i,j,r+1})\leq f_{i,j}(y^u_{i,j,n^u_{i,j}})=q^u_{i,j}$ for $r\in [1,n^u_{i,j}-1]$. Let $M_u=\sum^{s}_{k=1} M_k$, where $M_i=\sum^{m_i}_{k=1}q^u_{i,k}$ for $i\in[1,s]$.

For a fixed integer $i\in [1,t]$, a certain vertex of each connected bipartite $(p^v_{i,j},q^v_{i,j})$-graph $G^v_{i,j}$ is identified with $v_i\in Y$ into a vertex, $j\in [1,n_i]$. Let $V(G^v_{i,j})=X^v_{i,j}\cup Y^v_{i,j}$ be the bipartition of $V(G^v_{i,j})$, where $X^v_{i,j}=\{x^v_{i,j,l}:l\in [1,m^v_{i,j}]\}$ and $Y^v_{i,j}=\{y^v_{i,j,r}:r\in [1,n^v_{i,j}]\}$; $G^v_{i,j}$ has a set-ordered graceful labelliing $g_{i,j}$ such that $g_{i,j}(x^v_{i,j,1})=0$,
$g_{i,j}(x^v_{i,j,l})<g_{i,j}(x^v_{i,j,l+1})$ for $l\in [1,m^v_{i,j}-1]$, $g_{i,j}(x^v_{i,j,m^v_{i,j}})<g_{i,j}(y^v_{i,j,1})$, $g_{i,j}(y^v_{i,j,r})<g_{i,j}(y^v_{i,j,r+1})$ for $r\in [1,n^v_{i,j}-1]$, $g_{i,j}(y^v_{i,j,n^v_{i,j}})=q^v_{i,j}$. Let $N_v=\sum^{t}_{k=1} N_k$, where $N_i=\sum^{n_i}_{k=1}q^v_{i,k}$ for $i\in[1,t]$.

\vskip 0.2cm

It is allowed that some $m_i=0$, or $n_i=0$, or $m^u_{i,j}=0$, or $n^u_{i,j}=0$, or $m^v_{i,j}=0$, or $n^v_{i,j}=0$ in the following discussion. We define a new labelling $h$ and construct the desired graph $H^*$ thorough the following steps \emph{A}, \emph{B} and \emph{C}.

\vskip 0.2cm

\emph{A.} Labelling the edges of $G$ and $G^u_{i,j}$ for $j\in [1,m_i]$ and $i\in [1,s]$, and the edges of $G^v_{k,r}$ for $r\in [1,n_k]$ and $k\in [1,t]$.

\emph{A1.} We label the edges $x^u_{1,j,l}y^u_{1,j,r}$ of $G^u_{1,j}$ as:  $h(x^u_{1,1,l}y^u_{1,1,r})=f_{1,j}(x^u_{1,1,l}y^u_{1,1,r})$
for edges $x^u_{1,1,l}y^u_{1,1,r}\in E(G^u_{1,1})$, and we define the labels of edges $x^u_{1,j,l}y^u_{1,j,r}\in E(G^u_{1,j})$ with $j\in [2,m_{1}]$ as
\begin{equation}\label{eqa:edge-label-111}
h(x^u_{1,j,l}y^u_{1,j,r})=f_{1,j}(x^u_{1,j,l}y^u_{1,j,r})+\sum^{j-1}_{k=1}q^u_{1,k}.
\end{equation}
When $i\in[2,s]$, for edges $x^u_{i,1,l}y^u_{i,1,r}\in E(G^u_{i,1})$ we define
\begin{equation}\label{eqa:edge-label-222}
h(x^u_{i,1,l}y^u_{i,1,r})=f_{i,1}(x^u_{i,1,l}y^u_{i,1,r})+\sum^{i-1}_{k=1} M_k,
\end{equation}
and for edges $x^u_{i,j,l}y^u_{i,j,r}\in E(G^u_{i,j})$, $j\in [2,m_{i}]$, we set
\begin{equation}\label{eqa:edge-label-333}
h(x^u_{i,j,l}y^u_{i,j,r})=f_{i,j}(x^u_{i,j,l}y^u_{i,j,r})+\sum^{i-1}_{k=1} M_k+\sum^{j-1}_{k=1}q^u_{i,k}.
\end{equation}

\emph{A2.} Labelling the edges of  $G^v_{i,j}$ for $j\in [1,n_i]$ and $i\in [1,t]$. We label the edges $x^v_{t,n_t,l}y^v_{t,n_t,r}$ of $G^v_{t,n_t}$ by  $h(x^v_{t,n_t,l}y^v_{t,n_t,r})=g_{t,n_t}(x^v_{t,n_t,l}y^v_{t,n_t,r})+M_u$; and  for edges $x^v_{t,j,l}y^v_{t,j,r}\in E(G^v_{t,j})$, $j\in [1,n_{t}-1]$, we set
\begin{equation}\label{eqa:edge-label-444}
h(x^v_{t,j,l}y^v_{t,j,r})=g_{t,j}(x^v_{t,j,l}y^v_{t,j,r})+M_u+\sum^{n_{t}}_{k=j+1}q^v_{t,k}.
\end{equation}

For $i\in[1,t-1]$, we define
\begin{equation}\label{eqa:edge-label-555}
h(x^v_{i,n_{i},l}y^v_{i,n_{i},r})=g_{i,n_{i}}(x^v_{i,n_{i},l}y^v_{i,n_{i},r})+M_u+\sum^{t}_{k=i+1} N_k
\end{equation} for edges $x^v_{i,n_{i},l}y^v_{i,n_{i},r}\in E(G^v_{i,1})$, and set \begin{equation}\label{eqa:edge-label-666}
h(x^v_{i,j,l}y^v_{i,j,r})=g_{i,j}(x^v_{i,j,l}y^v_{i,j,r})+M_u+\sum^{t}_{k=i+1} N_k+\sum^{n_{i}}_{k=j+1}q^u_{i,k}
\end{equation} for edges $x^v_{i,j,l}y^v_{i,j,r}\in E(G^v_{i,j})$, $j\in [1,n_{i}-1]$.

\emph{A3.} We label the edges $u_iv_j$ of  $G$ as: $h(u_iv_j)=f(u_iv_j)+M_u+N_v$ for $u_iv_j\in E(G)$.

\vskip 0.2cm

\emph{B.} Labelling the vertices of $G$ and $G^u_{i,j}$ for $j\in [1,m_i]$ and $i\in [1,s]$, and  the vertices of $G^v_{k,r}$ for $r\in [1,n_k]$ and $k\in [1,t]$.

\emph{B1.} We label the vertices of $G$ as: $h(u_i)=f(u_i)$ for $u_i\in X$, and $h(v_j)=f(v_j)+M_u+N_v$ for $v_j\in Y$.

\emph{B2.} We label the vertices of $G^u_{i,j}$ in the following way:

$(1)$ $h(x^u_{1,j,1})=h(u_1)=0$ for $j\in [1,m_1]$;
$h(x^u_{1,j,l})=f_{1,j}(x^u_{1,j,l})$ for $l\in [2,m^u_{1,1}]$ and $j\in [1,m_1]$; $h(y^u_{1,1,r})=f_{1,j}(y^u_{1,1,r})$ for $r\in [1,n^u_{1,1}]$; and $h(y^u_{1,j,r})=f_{1,j}(y^u_{1,j,r})+\sum^{j-1}_{k=1} q^u_{1,k}$ for $r\in [1,n^u_{1,1}]$ and $j\in [2,m_1]$. Under the vertex labels we verify the edge labels
$${
\begin{split}
h(x^u_{1,1,l}y^u_{1,1,r})&=|h(x^u_{1,1,l})-h(y^u_{1,1,r})|=h(y^u_{1,1,r})-h(x^u_{1,1,l})\\
&=f_{1,1}(y^u_{1,1,r})-f_{1,1}(x^u_{1,1,l})=f_{1,1}(x^u_{1,1,l}y^u_{1,1,r})
\end{split}}
$$
for edges $x^u_{1,1,l}y^u_{1,1,r}\in E(G^u_{1,1})$; and we compute the labels of edges $x^u_{1,j,l}y^u_{1,j,r}\in E(G^u_{1,j})$ for $j\in [2,m_1]$ in the following
$${
\begin{split}
h(x^u_{1,j,l}y^u_{1,j,r})&=|h(x^u_{1,j,l})-h(y^u_{1,j,r})|=h(y^u_{1,j,r})-h(x^u_{1,j,l})\\
&=f_{1,j}(y^u_{1,j,r})+\sum^{j-1}_{k=1} q^u_{1,k}-f_{1,j}(x^u_{1,j,l})=f_{1,j}(x^u_{1,j,l}y^u_{1,j,r})+\sum^{j-1}_{k=1} q^u_{1,k},
\end{split}}
$$ which is equal to (\ref{eqa:edge-label-111}).

$(2)$ For $i\in [2,s]$, we define $h(x^u_{i,j,1})=f_{i,j}(x^u_{i,j,1})+h(u_i)=h(u_i)$ and
$h(x^u_{i,j,l})=f_{i,j}(x^u_{i,j,l})+h(u_i)$ for $l\in [2,m^u_{i,j}]$ and $j\in [1,m_i]$; $h(y^u_{i,1,r})=h(u_i)+f_{i,1}(y^u_{i,1,r})+\sum^{i-1}_{k=1} M_k$ for $r\in [1,n^u_{i,1}]$; and $h(y^u_{i,j,r})=f_{i,j}(y^u_{i,j,r})+\sum^{j-1}_{k=1} q^u_{i,k}+\sum^{i-1}_{k=1} M_k$ for $r\in [2,n^u_{i,j}]$ and $j\in [2,m_i]$. As verification of the edge labels, we can evaluate
\begin{equation}\label{eqa:c3333}
h(x^u_{i,1,l}y^u_{i,1,r})=|h(x^u_{i,1,l})-h(y^u_{i,1,r})|=h(y^u_{i,1,r})-h(x^u_{i,1,l})=f_{i,1}(x^u_{i,1,l}y^u_{i,1,r})+\sum^{i-1}_{k=1} M_k
\end{equation}
for edges $x^u_{i,1,l}y^u_{i,1,r}\in E(G^u_{i,1})$; and for edges $x^u_{i,j,l}y^u_{i,j,r}\in E(G^u_{i,j})$, $j\in [2,m_i]$,
\begin{equation}\label{eqa:c4444}
{
\begin{split}
h(x^u_{i,j,l}y^u_{i,j,r})&=|h(x^u_{i,j,l})-h(y^u_{i,j,r})|=h(y^u_{i,j,r})-h(x^u_{i,j,l})\\
&=f_{i,j}(y^u_{i,j,r})+\sum^{i-1}_{k=1} M_k+\sum^{j-1}_{k=1} q^u_{1,k}-f_{i,j}(x^u_{i,j,l})\\
&=f_{i,j}(x^u_{i,j,l}y^u_{i,j,r})+\sum^{i-1}_{k=1} M_k+\sum^{j-1}_{k=1} q^u_{1,k}.
\end{split}}
\end{equation}
The above (\ref{eqa:c3333}) and (\ref{eqa:c4444}) are equal to (\ref{eqa:edge-label-222}) and (\ref{eqa:edge-label-333}), respectively.

\emph{B3.} We label the vertices of $G^v_{i,j}$ in the following: $(1')$ $h(x^v_{t,j,1})=h(v_t)-g_{t,j}(x^v_{t,j,1})=h(v_t)$ for $j\in [1,n_t]$;
$h(x^v_{t,j,l})=h(v_t)-g_{t,j}(x^v_{t,j,l})$ for $l\in [2,n^v_{t,j}]$ and $j\in [1,n_t]$;
$h(y^v_{t,n_t,r})=h(v_t)-g_{t,n_t}(y^v_{t,n_t,r})-M_u$ for $r\in [1,n^v_{t,n_t}]$; and
$${
\begin{split}
h(y^v_{t,j,r})&=h(v_t)-g_{t,j}(y^v_{t,j,r})-M_u-\sum^{n_t}_{k=j+1} q^v_{t,k}\\
&=f(v_t)-g_{t,j}(y^v_{t,j,r})+N_v-\sum^{n_t}_{k=j+1} q^v_{t,k}\geq f(v_t)+\sum^{t-1}_{k=1}N_k
\end{split}}
$$ for $r\in [1,n^v_{t,j}]$ and $j\in [1,n_t-1]$. We verify the edge labels as follows:
$${
\begin{split}
h(x^v_{t,n_t,l}y^v_{t,n_t,r})&=|h(x^v_{t,n_t,l})-h(y^v_{t,n_t,r})|=h(x^v_{t,n_t,l})-h(y^v_{t,n_t,r})\\
&=g_{t,n_t}(y^v_{t,n_t,r})+M_u-g_{t,n_t}(x^v_{t,n_t,l})
=g_{t,n_t}(x^v_{t,n_t,l}y^v_{t,n_t,r})+M_u
\end{split}}
$$
for edges $x^v_{t,n_t,l}y^v_{t,n_t,r}\in E(G^v_{t,n_t})$; and for edges $x^v_{t,j,l}y^v_{t,j,r}\in E(G^v_{t,j})$ with $j\in [1,n_t-1]$, we have
$${
\begin{split}
h(x^v_{t,j,l}y^v_{t,j,r})&=h(x^v_{t,j,l})-h(y^v_{t,j,r})=g_{t,j}(y^v_{t,j,r})+M_u+\sum^{n_t}_{k=j+1} q^v_{t,k}-g_{t,j}(x^v_{t,j,l})\\
&=g_{t,j}(x^v_{t,j,l}y^v_{t,j,r})+M_u+\sum^{n_t}_{k=j+1} q^v_{t,k}.
\end{split}}
$$ The above edge labels are coincided with (\ref{eqa:edge-label-444}).

$(2')$ For $i\in [1,t-1]$, we set $h(x^v_{i,j,1})=h(v_i)-g_{i,j}(x^v_{i,j,1})=h(v_i)$,
$h(x^v_{i,j,l})=h(v_i)-g_{i,j}(x^v_{i,j,l})$ for $l\in [2,n^v_{i,j}]$ and $j\in [1,n_i]$; $h(y^v_{i,n_i,r})=h(v_i)-g_{i,n_i}(y^v_{i,n_i,r})-M_u-\sum^{t}_{k=i+1}N_{k}$ for $r\in [1,n^v_{i,n_i}]$; and
$${
\begin{split}
h(y^v_{i,j,r})&=h(v_i)-g_{i,j}(y^v_{i,j,r})-M_u-\sum^{t}_{k=i+1}N_{k}-\sum^{n_i}_{k=j+1} q^v_{i,k}\\
&=f(v_i)-g_{i,j}(y^v_{i,j,r})+\sum^{i}_{k=1}N_{k}-\sum^{n_i}_{k=j+1} q^v_{i,k}\geq f(v_i)+\sum^{i-1}_{k=1}N_{k}
\end{split}}
$$ for $r\in [1,n^v_{i,j}]$ and $j\in [1,n_i-1]$. We need to check the edge labels under the vertex labels of $G^v_{i,j}$. Considering edges $x^v_{i,j,l}y^v_{i,j,r}\in E(G^v_{i,j})$ with $j\in [1,n_t-1]$ and  $i\in [1,t-1]$, we can compute
$${
\begin{split}
h(x^v_{i,j,l}y^v_{i,j,r})&=|h(x^v_{i,j,l})-h(y^v_{i,j,r})|=h(x^v_{i,j,l})-h(y^v_{i,j,r})\\
&=g_{i,j}(x^v_{i,j,l}y^v_{i,j,r})+M_u+\sum^{t}_{k=i+1}N_{k}+\sum^{n_i}_{k=j+1} q^v_{i,k},
\end{split}}
$$
which is equal to (\ref{eqa:edge-label-666}); and the labels of edges $x^v_{i,n_i,l}y^v_{i,n_i,r}\in E(G^v_{i,n_i})$ are
$${
\begin{split}
h(x^v_{i,n_i,l}y^v_{i,n_i,r})&=|h(x^v_{i,n_i,l})-h(y^v_{i,n_i,r})|=h(x^v_{i,n_i,l})-h(y^v_{i,n_i,r})\\
&=g_{i,n_i}(x^v_{i,n_i,l}y^v_{i,n_i,r})+M_u+\sum^{t}_{k=i+1}N_{k},
\end{split}}
$$ which is equal to (\ref{eqa:edge-label-555}).

\emph{C.} We, now, construct the desired connected graph $H^*$ by identifying $x^u_{i,j,1}\in X^u_{i,j}\subset V(G^u_{i,j})$ with $u_i\in X\subset V(G)$ into a vertex for  $j\in [1,m_i]$ and $i\in [1,s]$, and then identifying $x^v_{a,b,1}\in X^v_{a,b}\subset V(G^v_{a,b})$ with $v_i\in Y\subset V(G)$ into a vertex for  $b\in [1,n_i]$ and $a\in [1,t]$, respectively.

\vskip 0.2cm

Let $e\,^*=|E(G)|+M_u+N_v$. Notice that $h(w)\in [h(u_1),h(v_t)]=[0,e\,^*]$ for $w\in V(H^*)$, and $$h(E(H^*))=h(E(G))+\sum^s_{i=1}\sum ^{m_i}_{j=1} h(E(G^u_{i,j}))+\sum^t_{k=1}\sum ^{n_k}_{r=1} h(E(G^v_{k,r}))=[1,e\,^*],$$ where $\sum ^{m_i}_{j=1} h(E(G^u_{i,j}))=0$ if $m_i=0$, and $\sum ^{n_k}_{r=1} h(E(G^v_{k,r}))=0$ if $n_k=0$.

Thereby, we conclude that $h$ is an in-improper graceful labelling of $H^*$.
\end{proof}

An example for illustrating Theorem  \ref{thm:1233333} is shown through Figure \ref{fig:example22} to Figure \ref{fig:example44}.

\begin{figure}[h]
\centering
\includegraphics[height=5cm]{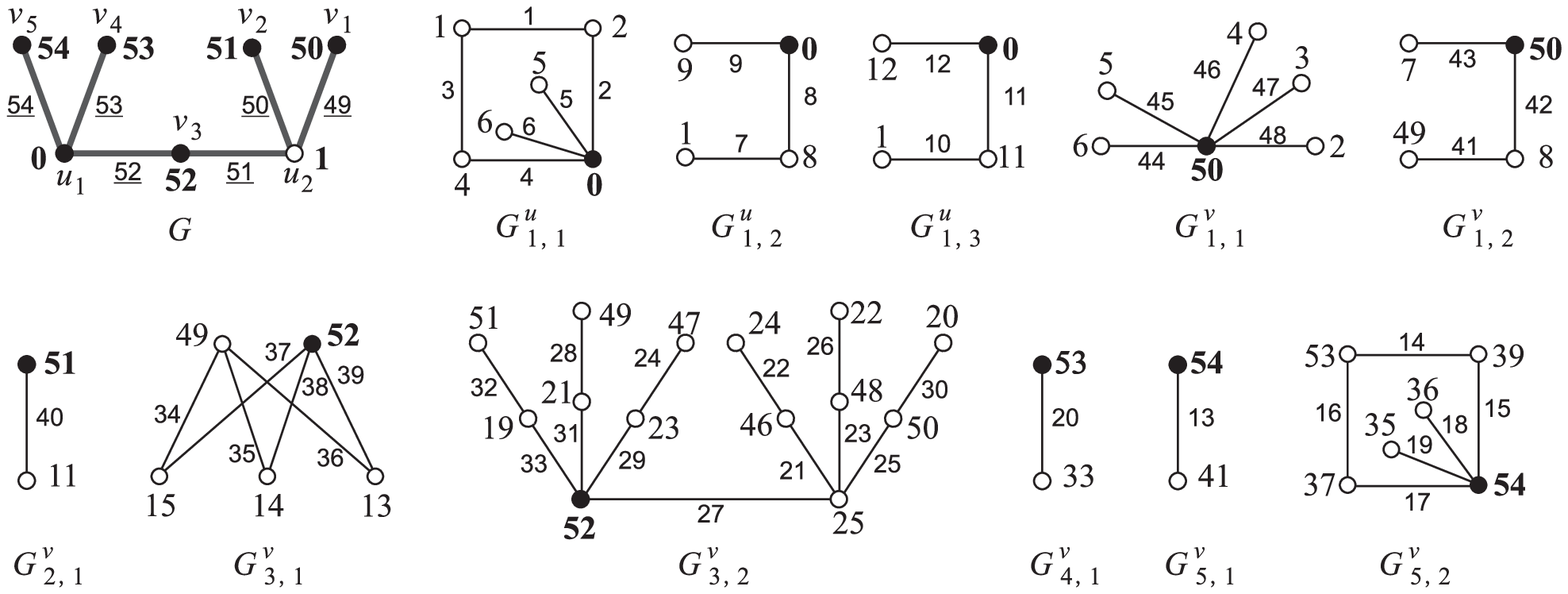}
\caption{\label{fig:example22}  {\small An example for illustrating the proof of Theorem \ref{thm:1233333}.}}
\end{figure}

\begin{figure}[h]
\centering
\includegraphics[height=4.8cm]{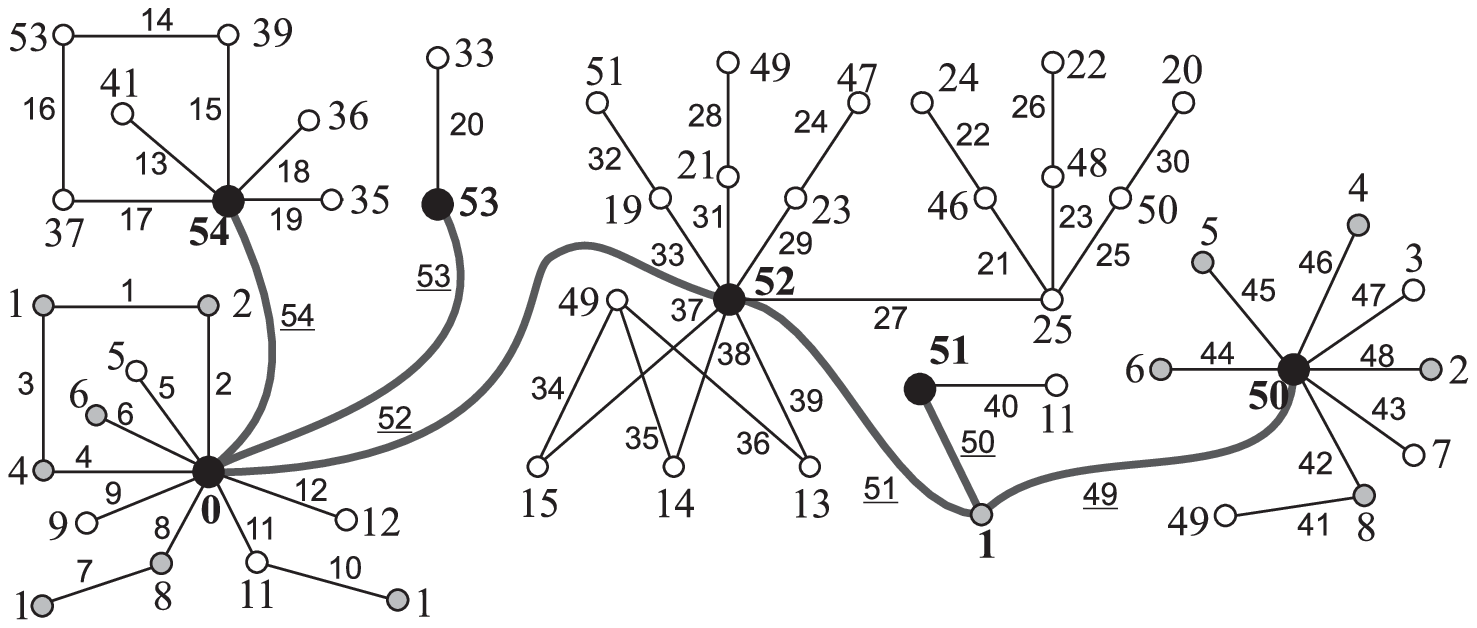}
\caption{\label{fig:example33}  {\small  The graph $H^*$ obtained in Theorem \ref{thm:1233333} has an in-improper graceful labelling.}}
\end{figure}

\begin{figure}[h]
\centering
\includegraphics[height=4.8cm]{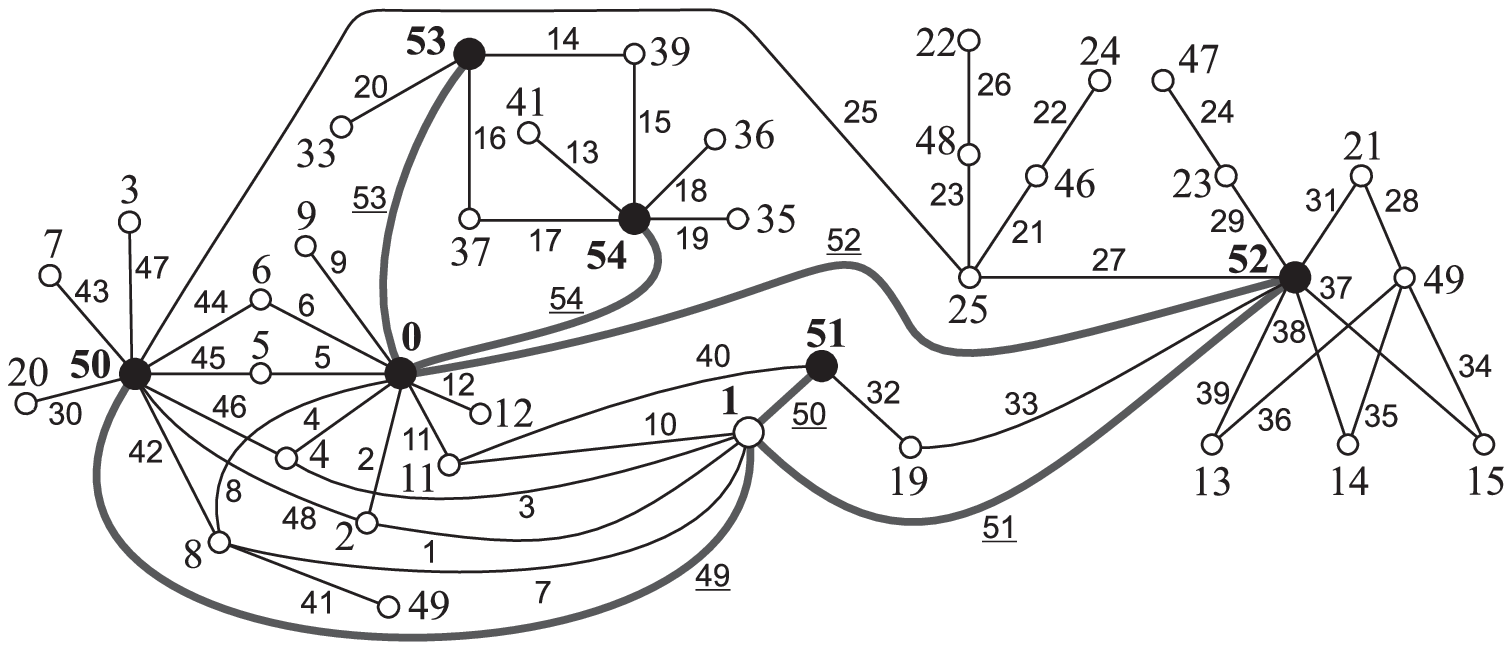}
\caption{\label{fig:example44}  {\small  A graceful graph obtained from $H^*$ shown in Figure \ref{fig:example33} by identifying those vertices having the same labels into one vertex.}}
\end{figure}

We conclude that identifying the vertices having the same labels in the graph $H^*$ shown in Theorem  \ref{thm:1233333} produces a connected graph admitting graceful/odd-graceful labellings (see an example shown in Figure \ref{fig:example44}). We, for further works, present an  extremal problem as follows.

\textbf{Problem.} \emph{Determine in-$(k^{\max}_{in},l^{\max}_{in})$-imgls, in-$(k^{\min}_{in},l^{\min}_{in})$-imgls and out-$(k^{\min}_{out}, l^{\min}_{out})$-imgls of trees}.

\vskip 0.6cm

\noindent \textbf{Acknowledgment.} The author, \emph{Jin Xu}, thanks the National Natural Science
Foundation of China Grant 61309015. The author, \emph{Bing Yao}, was supported by the National Natural Science
Foundation of China under grants No. 61163054 and No. 61363060.

\end{document}